\documentclass[a4paper, 11pt]{article}
\usepackage{amssymb}
\usepackage{amsmath}
\usepackage{amsthm}

\newcommand{\be}{\begin{equation}}
\newcommand{\ee}{\end{equation}}
\newcommand{\bea}{\begin{eqnarray}}
\newcommand{\eea}{\end{eqnarray}}

\newcommand{\rr}{\mathrm{r}}
\newcommand{\ri}{\mathrm{i}}

\newcommand{\bN}{\mathbb{N}}
\newcommand{\bR}{\mathbb{R}}
\newcommand{\bC}{\mathbb{C}}

\newcommand{\cB}{\mathcal{B}}

\newcommand{\cD}{\mathcal{D}}

\newcommand{\cP}{\mathcal{P}}

\newcommand{\cO}{\mathcal{O}}
\newcommand{\cV}{\mathcal{V}}

\newcommand{\cA}{\mathcal{A}}
\newcommand{\cM}{\mathcal{M}}

\newcommand{\calR}{\mathcal{R}}

\newcommand{\mfa}{\mathfrak{a}}
\newcommand{\mfc}{\mathfrak{c}}
\newcommand{\mfm}{\mathfrak{m}}
\newcommand{\mfu}{\mathfrak{u}}
\newcommand{\mfg}{\mathfrak{g}}
\newcommand{\mfgl}{\mathfrak{gl}}
\newcommand{\mfk}{\mathfrak{k}}
\newcommand{\mfp}{\mathfrak{p}}
\newcommand{\mfX}{\mathfrak{X}}

\newcommand{\mfL}{\mathfrak{L}}

\newcommand{\bsone}{\boldsymbol{1}}

\newcommand{\bsX}{\boldsymbol{X}}
\newcommand{\bsY}{\boldsymbol{Y}}

\newcommand{\bsC}{\boldsymbol{C}}

\newcommand{\diag}{\text{diag}}

\newcommand{\ad}{\mathrm{ad}}
\newcommand{\wad}{\widetilde{\ad}}
\newcommand{\Ad}{\mathrm{Ad}}
\newcommand{\tr}{\mathrm{tr}}
\newcommand{\Id}{\mathrm{Id}}

\newcommand{\reg}{\mathrm{reg}}

\newcommand{\ran}{\mathrm{ran}}
\newcommand{\red}{\mathrm{red}}

\newcommand{\IM}{\mathrm{Im}}
\newcommand{\RE}{\mathrm{Re}}
\newcommand{\End}{\mathrm{End}}

\newcommand{\ddd}{\mathrm{d}}
\newcommand{\ext}{\mathrm{ext}}

\newcommand{\inner}{\, \lrcorner \,}

\newcommand{\half}{\frac{1}{2}}
\newcommand{\quarter}{\frac{1}{4}}

\renewcommand{\theequation}
{\arabic{section}.\arabic{equation}}

\theoremstyle{plain}
\newtheorem{THEOREM}{Theorem}
\newtheorem{LEMMA}[THEOREM]{Lemma}

\begin{document}
\begin{center}
\Large{\textbf{
On the $r$-matrix structure of the hyperbolic $BC_n$ Sutherland model
}}
\end{center}
\bigskip
\begin{center}
B.G.~Pusztai\\
Bolyai Institute, University of Szeged,\\
Aradi v\'ertan\'uk tere 1, H-6720 Szeged, 
Hungary\\
e-mail: \texttt{gpusztai@math.u-szeged.hu}
\end{center}
\bigskip
\begin{abstract}
Working in a symplectic reduction framework, we construct 
a dynamical $r$-matrix for the classical hyperbolic $BC_n$ Sutherland 
model with \emph{three independent coupling constants}. We also examine 
the Lax representation of the dynamics and its equivalence with the 
Hamiltonian equation of motion.



\end{abstract}
\newpage

\section{Introduction}
\label{S1}
\setcounter{equation}{0}
The Calogero--Moser--Sutherland-type many-particle models are intensively
studied integrable systems with deep connections to many
important branches of mathematics and physics. As a classical Hamiltonian 
system, the hyperbolic $BC_n$ Sutherland model is defined on the cotangent 
bundle of the open subset
\be
	\mfc 
	= \{ q = (q_1, \ldots, q_n) \in \bR^n 
		\, | \, 
		q_1 > \ldots > q_n > 0 \} \subset \bR^n.
\label{mfc}
\ee
For convenience we identify the phase space $T^* \mfc$ with the product 
manifold
\be
	\cP^S = \mfc \times \bR^n 
	= \{ (q, p) \, | \, q \in \mfc, \, p \in \bR^n \},
\label{cP_S}
\ee
endowed with the standard symplectic form
\be
	\omega^S = \sum_{c = 1}^n \ddd q_c \wedge \ddd p_c.
\label{omega_S}
\ee
The dynamics is governed by the interacting many-body Hamiltonian
\be
\begin{split}
	H^S = & \, \half \sum_{c = 1}^n p_c^2
		+ \sum_{c = 1}^n
			\left(
				\frac{g_1^2}{\sinh^2(q_c)} 
				+ \frac{g_2^2}{\sinh^2(2 q_c)} 
			\right)	\\
	& + \sum_{1 \leq a < b \leq n} 
		\left(
			\frac{g^2}{\sinh^2(q_a - q_b)} 
			+ \frac{g^2}{\sinh^2(q_a + q_b)} 
		\right)
\end{split}
\label{H_S}
\ee
with coupling constants $g^2, g_1^2, g_2^2 \in \bR$ satisfying
$g^2 > 0$, $g_2^2 \geq 0$ and $g_1^2 > -\quarter g_2^2$. 

By applying the projection method on the geodesic system of the non-compact
Riemannian symmetric space $SU(n + 1, n) / S(U(n + 1) \times U(n))$, 
Olshanetsky and Perelomov constructed a Lax representation of the $BC_n$ 
Sutherland dynamics and analyzed the issue of solvability as well, but 
only under the restrictive assumption $g_1^2 - 2 g^2 + \sqrt{2} g g_2 = 0$ 
(for details see e.g.
\cite{OlshaPere76}, 
\cite{OlshaPere}, 
\cite{PerelomovBook}).
As the algebraic methods prevailed, the Lax representation of the dynamics
was soon established for arbitrary values of the coupling constants
(see e.g. 
\cite{Inozemtsev_Meshcheryakov_1985},
\cite{Opdam_1988},
\cite{Bordner_Sasaki_Takasaki},
\cite{Bordner_Corrigan_Sasaki}).
Somewhat surprisingly, the symplectic reduction derivation of the $BC_n$ 
Sutherland model with three \emph{independent} coupling constants is only 
a relatively recent development \cite{Feher_Pusztai_2007}. Besides 
providing a nice geometric picture and an efficient solution algorithm, 
the symplectic reduction approach has also allowed us to construct 
action-angle variables to the $BC_n$ Sutherland model and to establish 
its duality with the rational $BC_n$ Ruijsenaars--Schneider--van Diejen 
system (see \cite{Pusztai_NPB2012}). Sticking to the powerful machinery of 
symplectic reduction, in this paper we construct a dynamical $r$-matrix 
for the $BC_n$ Sutherland model. By accomplishing this task we generalize 
the results of Avan, Babelon and Talon on the $r$-matrix structure of the 
hyperbolic $C_n$ Sutherland model, which appeared in their paper 
\cite{Avan_Babelon_Talon_1994}.

The rest of the paper is organized as follows. In order to keep the 
presentation self-contained, in the next section we provide a brief 
account on the group theoretic and symplectic geometric background 
underlying the symplectic reduction derivation of the hyperbolic $BC_n$
Sutherland model. Built upon the reduction approach outlined 
in Section \ref{S2}, in Section \ref{S3} we construct a $q$-dependent
dynamical $r$-matrix for the most general hyperbolic $BC_n$ Sutherland 
model with three independent coupling constants. The new results 
are summarized concisely in Theorems \ref{theorem_r_matrix} and 
\ref{theorem_Lax_eqn}. Subsequently, in Section \ref{S4}, we offer a 
short discussion on possible applications and related open problems.
Finally, some auxiliary material on the Lie algebra $\mfu(n, n)$
can be found in an appendix.

\section{Preliminaries}
\label{S2}
\setcounter{equation}{0}
In this section we review the symplectic reduction 
derivation of the hyperbolic $BC_n$ Sutherland model. 
For convenience, we closely follow the ideas and 
conventions presented in \cite{Pusztai_NPB2012}. 

\subsection{Group theoretic background}
Take an arbitrary positive integer $n \in \bN$, let $N = 2 n$, and consider 
the $N \times N$ matrix
\be
	\bsC = \begin{bmatrix}
	0_n & \bsone_n \\
	\bsone_n & 0_n
	\end{bmatrix}.
\label{bsC}
\ee
The matrix Lie group
\be
	G = \{ y \in GL(N, \bC) \, | \, y^* \bsC y = \bsC \}
\label{G}
\ee
provides an appropriate model of the real reductive Lie group $U(n, n)$. 
Its Lie algebra 
\be
	\mfg = \mfu(u, n) 
	= \{ Y \in \mfgl(N, \bC) \, | \, Y^* \bsC + \bsC Y = 0 \}
\label{mfg}
\ee
comes naturally equipped with the $\Ad$-invariant symmetric bilinear form
\be
	\langle \, , \rangle \colon \mfg \times \mfg \rightarrow \bR,
	\quad
	(Y, \tilde{Y}) \mapsto \langle Y, \tilde{Y} \rangle 
	= \tr(Y \tilde{Y}).
\label{bilinear}
\ee
The fixed-point set of the Cartan involution 
$\Theta(y) = (y^{-1})^*$ $(y \in G)$ can be identified as
\be
	K = \{ y \in G \, | \, \Theta(y) = y \}
	\cong U(n) \times U(n),
\ee 
meanwhile the corresponding Lie algebra involution 
$\theta(Y) = - Y^*$ $(Y \in \mfg)$ naturally induces the Cartan 
decomposition $\mfg = \mfk \oplus \mfp$ with the eigenspaces
\be
	\mfk = \ker(\theta - \Id_\mfg)
	\quad \text{and} \quad
	\mfp = \ker(\theta + \Id_\mfg).
\ee
That is, each $Y \in \mfg$ can be uniquely decomposed as $Y = Y_+ + Y_-$
with $Y_+ \in \mfk$ and $Y_- \in \mfp$. Note that the bilinear form 
(\ref{bilinear}) is negative definite on the subalgebra $\mfk$, whereas 
it is positive definite on the complementary subspace $\mfp$.

Now, with each real $n$-tuple $q = (q_1, \ldots, q_n) \in \bR^n$ we 
associate the $N \times N$ diagonal matrix
\be
	Q = \diag(q_1, \ldots, q_n, -q_1, \ldots, -q_n) \in \mfp.
\label{Q}
\ee
Clearly the subset $\mfa = \{ Q \in \mfp \, | \, q \in \bR^n \}$ is a 
maximal Abelian subspace in $\mfp$, which can be naturally identified
with $\bR^n$. Under the adjoint action, the 
centralizer of $\mfa$ in $K$ is the subgroup
\be
	M = Z_K (\mfa)
	= \{ 
		\diag(e^{\ri \chi_1}, \ldots, e^{\ri \chi_n}, 
			e^{\ri \chi_1}, \ldots, e^{\ri \chi_n}) 
		\, | \, \chi_1, \ldots, \chi_n \in \bR 
	\} \subset K
\label{M}
\ee
with Lie algebra
\be
	\mfm 
	= \{ 
		\diag(\ri \chi_1, \ldots, \ri \chi_n, 
			\ri \chi_1, \ldots, \ri \chi_n)
		\, | \, \chi_1, \ldots, \chi_n \in \bR 
	\} \subset \mfk.
\label{mfm}
\ee
Let $\mfa^\perp$ (respectively $\mfm^\perp$) denote the set of the 
off-diagonal elements of $\mfp$ (respectively $\mfk$); then with
respect to the bilinear form (\ref{bilinear}) we have
the refined \emph{orthogonal} decomposition
\be
	\mfg = \mfm \oplus \mfm^\perp \oplus \mfa \oplus \mfa^\perp.
\label{mfg_refined_decomposition}
\ee
Practically, each Lie algebra element $Y \in \mfg$ can be decomposed as 
\be
	Y = Y_\mfm + Y_{\mfm^\perp} + Y_\mfa + Y_{\mfa^\perp}
\ee 
with unique components belonging to the subspaces indicated by the
subscripts.

Notice that for each $q \in \bR^n$ the operator 
$\ad_Q = [Q, \cdot] \in \mfgl(\mfg)$ leaves the subspace 
$\mfm^\perp \oplus \mfa^\perp$ invariant; therefore the restricted
operator
\be
	\wad_Q = \ad_Q |_{\mfm^\perp \oplus \mfa^\perp} 
	\in \mfgl(\mfm^\perp \oplus \mfa^\perp)
\label{wad_Q}
\ee
is well defined. Recall that the regular part of $\mfa$ consists 
of those diagonal matrices $Q$, for which the linear operator $\wad_Q$ 
is invertible. Clearly the standard Weyl chamber
\be
	\{ Q \in \mfa \, | \, q_1 > \ldots > q_n > 0 \}
\label{Weyl_chamber}
\ee
is a connected component of the regular part of $\mfa$. 
For simplicity, in the rest of the paper we shall identify this Weyl 
chamber with the configuration space $\mfc$ (\ref{mfc}). 

In deriving the Sutherland model from symplectic reduction, the so-called
$K A K$ decomposition of $G$ plays a crucial role. It basically says that
the map
\be
	\bR^n \times K \times K \ni (q, k_L, k_R) 
	\mapsto 
	k_L e^Q k_R^{-1} \in G
\label{KAK}
\ee
is \emph{onto}. Let $G_\reg$ denote the image of $\mfc \times K \times K$
under the above map. As is known, the subset of regular elements, 
$G_\reg$, is an open and dense submanifold of $G$. Moreover, the 
smooth map
\be
	\pi \colon \mfc \times K \times K 
	\twoheadrightarrow G_\reg,
	\quad
	(q, k_L, k_R) \mapsto k_L e^Q k_R^{-1}
	\label{pi}
\ee
is a principal $M$-bundle in a natural manner. Consequently, we arrive 
at the natural identification $G_\reg \cong \mfc \times (K \times K) / M_*$,
where $M_*$ stands for the diagonal embedding of $M$ (\ref{M}) into the 
product Lie group $K \times K$. That is, $M_*$ consists of the pairs 
$(m, m) \in K \times K$ with $m \in M$.

We conclude this subsection with a brief excursion on certain adjoint
orbit of $\mfk$, which is at the heart of the symplectic reduction
derivation of the $BC_n$ Sutherland model. For this, let us consider the 
following set of column vectors
\be
	S = \{ V \in \bC^N \, | \, \bsC V + V = 0, \, V^* V = N \},
\label{S}
\ee
which can be seen as a sphere of real dimension $2 n - 1$. With each
vector $V \in S$ we associate the matrix
\be
	\xi(V) 
	= \ri \mu (V V^* - \bsone_N) + \ri (\mu - \nu) \bsC \in \mfk,
\label{xi_def}
\ee
where $\mu, \nu \in \bR \setminus \{ 0 \}$ are arbitrary non-zero real 
parameters.
Let us now introduce the distinguished column vector $E \in S$ with
components
\be
	E_a = - E_{n + a} = 1
	\quad
	(a \in \bN_n = \{ 1, \ldots, n \}).
\label{E}
\ee
Notice that the adjoint orbit in $\mfk$ passing through the element
$\xi(E)$ has the form
\be
	\cO = \cO(\xi(E)) = \{ \xi(V) \in \mfk \, | \, V \in S \}.
\label{cO}
\ee
More precisely, with the free action 
$U(1) \times S \ni (e^{\ri \psi}, V) \mapsto e^{\ri \psi} V \in S$, 
the map 
\be
	\xi \colon S \twoheadrightarrow \cO,
	\quad
	V \mapsto \xi(V)
\label{xi}
\ee
is a principal $U(1)$-bundle. Therefore the identification 
$\cO \cong S / U(1)$ is immediate.

\subsection{The Sutherland model from symplectic reduction}
We continue with a short review on the symplectic structure of the 
cotangent bundle of $G$. For convenience, we trivialize this bundle 
by left translations. Therefore, making use of the linear isomorphism 
$\mfg^* \cong \mfg$ induced by (\ref{bilinear}), we can think of the 
product manifold $\cP = G \times \mfg$ as an appropriate model of 
$T^* G$. At each point $(y, Y) \in \cP$ the canonical symplectic 
form $\omega \in \Omega^2(\cP)$ can be written as
\be
	\omega_{(y, Y)}(\Delta y \oplus \Delta Y, \delta y \oplus \delta Y)
	= \langle y^{-1} \Delta y, \delta Y \rangle
		- \langle y^{-1} \delta y, \Delta Y \rangle
		+ \langle [y^{-1} \Delta y, y^{-1} \delta y], Y \rangle,
	\label{omega}
\ee
where $\Delta y \oplus \Delta Y$ and $\delta y \oplus \delta Y$ 
are arbitrary elements belonging to the tangent space
$T_y G \oplus \mfg \cong T_{(y, Y)} \cP$. Turning to the adjoint orbit 
$\cO$ (\ref{cO}), remember that it also carries a natural symplectic 
structure induced by the Kirillov--Kostant--Souriau symplectic form
$\omega^\cO \in \Omega^2(\cO)$. Let us keep in mind that at each point 
$\rho \in \cO$ it takes the form
\be
	\omega^\cO_\rho ([X, \rho], [Z, \rho]) 
	= \langle \rho, [X, Z] \rangle,
	\label{omega_cO}
\ee
where $[X, \rho], [Z, \rho] \in T_\rho \cO$ are arbitrary tangent 
vectors with $X, Z \in \mfk$. Now, motivated by the standard 
`shifting trick' of symplectic reduction, we introduce the product 
symplectic manifold
\be
	(\cP^\ext, \omega^\ext) = (\cP \times \cO, \omega + \omega^\cO).
	\label{cP_ext}
\ee
For an arbitrary function $F \in C^\infty(\cP^\ext)$, at each 
$u = (y, Y, \rho) \in \cP^\ext$, we define its gradients
\be
	\nabla^G F(u) \in \mfg,
	\quad
	\nabla^\mfg F(u) \in \mfg,
	\quad
	\nabla^\cO F(u) \in T_\rho \cO \subset \mfk,
	\label{gradient}
\ee
by requiring
\be
	(\ddd F)_u (\delta y \oplus \delta Y \oplus [X, \rho])
	= \langle \nabla^G F(u), y^{-1} \delta y \rangle
		+ \langle \nabla^\mfg F(u), \delta Y \rangle
		+ \langle \nabla^\cO F(u), X \rangle
	\label{gradient_def}
\ee
for all $\delta y \in T_y G$, $\delta Y \in \mfg$ and $X \in \mfk$.
By combining the definition $\bsX_F \inner \omega^\ext = \ddd F$ with 
the above formula, for the Hamiltonian vector field 
$\bsX_F \in \mfX(\cP^\ext)$ we find
\be
	(\bsX_F)_u
	= (y \nabla^\mfg F(u))_y 
		\oplus ([Y, \nabla^\mfg F(u)] - \nabla^G F(u))_Y
		\oplus (-\nabla^\cO F(u))_\rho.
	\label{bsX_u}
\ee
Consequently, from the definition 
$\{ F, H \}^\ext = \omega^\ext (\bsX_F, \bsX_H)$
it is immediate that the Poisson bracket of any pair of functions
$F, H \in C^\infty(\cP^\ext)$ takes the form
\be
\begin{split}
	\{ F, H \}^\ext(u)
	& = \langle \nabla^G F(u), \nabla^\mfg H(u) \rangle
		- \langle \nabla^G H(u), \nabla^\mfg F(u) \rangle \\
    	& \quad - \langle [\nabla^\mfg F(u), \nabla^\mfg H(u)], Y \rangle
		+ \omega^\cO_\rho (\nabla^\cO F(u), \nabla^\cO H(u)). 
\end{split}
\label{PB_ext}
\ee

Inspired by the $K A K$ decomposition of $G$ (\ref{KAK}), let us 
observe that the map
\be
	\Phi^\ext \colon (K \times K) \times \cP^\ext \rightarrow \cP^\ext,
	\quad
	((k_L , k_R), (y, Y, \rho)) 
	\mapsto 
	(k_L y k_R^{-1}, k_R Y k_R^{-1}, k_L \rho k_L^{-1})  
	\label{Phi_ext}
\ee
is a symplectic left action of $K \times K$ on $\cP^\ext$, 
admitting a $K \times K$-equivariant momentum map
\be
	J^\ext \colon \cP^\ext \rightarrow \mfk \oplus \mfk,
	\quad
	(y, Y, \rho)
	\mapsto
	( (y Y y^{-1})_+ + \rho ) \oplus (- Y_+ - \kappa \ri \bsC)
	\label{J_ext}
\ee
for all $\kappa \in \bR$. As is known (see \cite{Feher_Pusztai_2007},
\cite{Pusztai_NPB2012}), the phase space of the hyperbolic $BC_n$ 
Sutherland model can be derived by reducing the extended phase space 
$\cP^\ext$ at the \emph{zero value} of the momentum map $J^\ext$. 
In the following we briefly summarize the main steps 
of the reduction. 

First, let $\mfL_0$ denote the set of those points $u$
of the extended phase space (\ref{cP_ext}), for which we have 
$J^\ext(u) = 0$. Note that the level set $\mfL_0$ turns out to be an 
embedded submanifold of $\cP^\ext$. To analyze its finer structure,
we introduce the Lax matrix
\be
	L \colon \cP^S \rightarrow \mfg,
	\quad
	(q, p) \mapsto L(q, p) = L_\mfp(q, p) - \kappa \ri \bsC,
	\label{L}
\ee
where $L_\mfp(q, p) \in \mfp$ is an Hermitian matrix having the 
block matrix structure
\be
	L_\mfp =
	\begin{bmatrix}
	 	\cA & \cB \\
	 	-\cB & -\cA
	\end{bmatrix}.
\ee
More precisely, the entries of the $n \times n$ matrices $\cA$ and $\cB$ 
are defined by the formulae
\be
	\cA_{a, b} = \frac{-\ri \mu}{\sinh(q_a - q_b)}, 
	\quad
	\cA_{c, c} = p_c, 
	\quad
	\cB_{a, b} = \frac{\ri \mu}{\sinh(q_a + q_b)}, 
	\quad
	\cB_{c, c} 
	= \ri \frac{\nu + \kappa \cosh(2 q_c)}{\sinh(2 q_c)}, 
\label{A_and_B} 
\ee
where $a, b, c \in \bN_n$, $a \neq b$. We also need the manifold
$\cM^S = \cP^S \times (K \times K) / U(1)_*$, where $U(1)_*$ denotes 
the diagonal embedding of $U(1)$ into $K \times K$. Now, one can 
verify that the map
\be
	\Upsilon^S \colon \cM^S \rightarrow \cP^\ext,
	\quad
	(q, p, (\eta_L, \eta_R) U(1)_*)
	\mapsto
	(\eta_L e^Q \eta_R^{-1}, 
	\eta_R L(q, p) \eta_R^{-1}, 
	\eta_L \xi(E) \eta_L^{-1})
	\label{Upsilon_S}
\ee
is an injective immersion with image $\Upsilon^S (\cM^S) = \mfL_0$.
Since the manifolds $\cM^S$ and $\mfL_0$ are of the same dimension, 
this observation leads to the identification $\mfL_0 \cong \cM^S$.

Second, by examining the (residual) action of $K \times K$ on the model 
space $\cM^S$ of the level set $\mfL_0$,  it is immediate that the 
base manifold of the trivial principal $(K \times K)/U(1)_*$-bundle
\be
	\pi^S \colon \cM^S \twoheadrightarrow \cP^S,
	\quad
	(q, p, (\eta_L, \eta_R) U(1)_*)
	\mapsto (q, p)
	\label{pi_S}
\ee
provides a convenient model for the reduced symplectic manifold. That 
is, we end up with the natural identifications
\be
	\cP^\ext /\! /_0 (K \times K)
	\cong \cM^S / (K \times K) 
	\cong \cP^S.
\ee
Making use of the defining relationship
$(\pi^S)^* \omega^\red = (\Upsilon^S)^* \omega^\ext$, for the reduced 
symplectic form we find immediately that $\omega^\red = 2 \omega^S$ with 
the canonical symplectic form $\omega^S$ (\ref{omega_S}). Consequently, 
for the reduced Poisson bracket we obtain
\be
	\{ \cdot \, , \cdot \}^S = 2 \{ \cdot \, , \cdot \}^\red.
	\label{PB_S_vs_PB_red}
\ee

Finally, let us consider the $K \times K$-invariant quadratic Hamiltonian
\be
	F_2 (y, Y, \rho) = \frac{1}{4} \langle Y, Y \rangle 
	= \frac{1}{4} \tr(Y^2)
	\qquad
	((y, Y, \rho) \in \cP^\ext).
\label{F_2}
\ee
It is clear that $F_2$ generates the `free' geodesic motion on the group 
manifold $G$. Due to its invariance, it survives the reduction and the 
corresponding reduced Hamiltonian coincides with the Hamiltonian of the 
Sutherland model (\ref{H_S}) with coupling constants
\be
	g^2 = \mu^2, 
	\quad 
	g_1^2 = \half \nu \kappa, 
	\quad 
	g_2^2 = \half (\nu - \kappa)^2.
	\label{coupling_relations}
\ee
Just now can we really appreciate the inclusion of the innocent
looking $\kappa$-dependent central element $\kappa \ri \bsC$
into the momentum map $J^\ext$ (\ref{J_ext}). Indeed, by specializing 
the parameters $(\mu, \nu, \kappa)$ appropriately, from the proposed 
reduction picture we can recover the most general hyperbolic $BC_n$ 
Sutherland model with \emph{three} independent coupling constants.  

\section{Construction of the $r$-matrix}
\label{S3}
\setcounter{equation}{0}
Since the eigenvalues of $L$ (\ref{L}) are in involution 
(see \cite{Pusztai_NPB2012}), we know from general principles that
the Lax matrix obeys an $r$-matrix Poisson bracket. As is known
(see e.g. \cite{Braden_1997}), there is a general, purely algebraic 
approach to find an explicit formula for the $r$-matrix. However, 
we rather follow the symplectic reduction approach put forward by 
Avan, Babelon and Talon in \cite{Avan_Babelon_Talon_1994}. 
It is worth mentioning that this geometric approach was later 
generalized and systematically exploited in \cite{Braden_et_al_2003}, 
leading to a uniform treatment of the classical $r$-matrix structure 
for various integrable systems.  

\subsection{Local extensions of the Lax matrix}
Take an arbitrary point $(q, p)$ of $\cP^S$ and 
keep it fixed. Notice that the point 
\be
	z_0 = (q, p, (\bsone_N, \bsone_N) U(1)_*) \in \cM^S
	\label{z_0}
\ee
projects onto $(q, p)$, i.e. $\pi^S(z_0) = (q, p)$. Now, pick an arbitrary 
function $f \in C^\infty(\cP^S)$. We say that a smooth function
\be
	\tilde{f} \colon U \rightarrow \bR
\ee 
defined on some open neighborhood $U \subset \cP^\ext$ of point
\be
	u_0 = \Upsilon^S (z_0) = (e^Q, L(q, p), \xi(E)) \in \cP^\ext
	\label{u_0}
\ee
is a \emph{local extension of $f$ around $u_0$}, if
\be
	\tilde{f} \circ \Upsilon^S \big{|}_{(\Upsilon^S)^{-1}(U)}
	= f \circ \pi^S \big{|}_{(\Upsilon^S)^{-1}(U)}.
	\label{local_extension_def}
\ee
As is known, this special class of local extensions can 
be used effectively to compute reduced Poisson brackets by evaluating 
certain `unreduced' Poisson brackets. More precisely, if 
$\tilde{f}, \tilde{g} \in C^\infty(U)$ are arbitrary local 
extensions of functions $f, g \in C^\infty(\cP^S)$ around $u_0$
in the sense of (\ref{local_extension_def}), then 
\be
	\{ f, g \}^\red(q, p) 
	= \{ \tilde{f}, \tilde{g} \}^\ext(u_0).
	\label{red_PB_vs_ext_PB}
\ee
In particular, in the following we shall make use of the above 
formula\footnote{
Consistently with the Dirac bracket, generally there are also some 
correction terms on the right hand side of (\ref{red_PB_vs_ext_PB}). 
However, since we reduce at the \emph{zero value} of the (equivariant) 
momentum map, and since by (\ref{local_extension_def}) our local 
extensions are (locally) $K \times K$-invariant on the level set 
$\mfL_0$, these corrections cancel. For details see e.g. Chapter 14 
in \cite{Babelon_Bernard_Talon_book}.
} 
to find an explicit expression for the $r$-matrix of the Sutherland 
model. The auxiliary functions defined below play an important role in 
constructing local extensions of the Lax matrix $L$ (\ref{L}).

We start with the study of the smooth principal $M$-bundle $\pi$ 
(\ref{pi}) induced by the $K A K$ decomposition of $G$. Since 
$\pi(q, \bsone_N, \bsone_N) = e^Q$, there is a smooth local section
\be
	\check{G} \ni y 
	\mapsto
	(\sigma_\mfc(y), \sigma_L(y), \sigma_R(y)) \in \mfc \times K \times K
	\label{sigma_section}
\ee
of $\pi$, defined on some small open neighborhood 
$\check{G} \subset G_\reg$ of $e^Q$, such that
\be
	(\sigma_\mfc(e^Q), \sigma_L(e^Q), \sigma_R(e^Q))
	= (q, \bsone_N, \bsone_N).
	\label{sigma_normalization}
\ee
Besides the above normalization, we may impose certain conditions
on the derivative of section (\ref{sigma_section}) at $e^Q$, too.
Notice that the tangent space of $\mfc \times K \times K$ at
$(q, \bsone_N, \bsone_N)$ can be identified as
\be
	T_{(q, \bsone_N, \bsone_N)}(\mfc \times K \times K)
	\cong T_q \mfc \oplus T_{\bsone_N} K \oplus T_{\bsone_N} K
	\cong \bR^n \oplus \mfk \oplus \mfk,
\ee
in which the vertical subspace of $\pi$ takes the form
\be
	\ker((\ddd \pi)_{(q, \bsone_N, \bsone_N)})
	= \{ 0 \oplus X \oplus X \in \bR^n \oplus \mfk \oplus \mfk
	\, | \,
	X \in \mfm \} \cong \mfm.
	\label{sigma_vertical_subspace}
\ee
Since $\bR^n \oplus \mfm^\perp \oplus \mfk$ is clearly a complementary
subspace of the vertical subspace, we may assume that at point $e^Q$ 
the derivative of (\ref{sigma_section}) maps into the complementary
`horizontal' subspace, i.e.
\be
	\ran((\ddd(\sigma_\mfc, \sigma_L, \sigma_R))_{e^Q})
	= \bR^n \oplus \mfm^\perp \oplus \mfk.
\ee
That is, we may assume that
\be
	\ran((\ddd \sigma_L)_{e^Q}) = \mfm^\perp.
	\label{sigma_derivative_condition}
\ee
In the following we will need an explicit formula for the 
derivative of (\ref{sigma_section}) at point $e^Q$.

\begin{LEMMA}
\label{lemma_sigma_derivative}
Under assumption (\ref{sigma_derivative_condition}), for each tangent
vector $\delta Y \in \mfg \cong T_{\bsone_N} G$ we have
\begin{align}
	& (\ddd \sigma_L)_{e^Q} (e^Q \delta Y) 
	= -\sinh(\wad_Q)^{-1} (\delta Y)_{\mfa^\perp}, \\
	& (\ddd \sigma_R)_{e^Q} (e^Q \delta Y)
	= -(\delta Y)_+ -\coth(\wad_Q) (\delta Y)_{\mfa^\perp}.
\end{align}
\end{LEMMA}

\begin{proof}
For simplicity, let us introduce the shorthand notations
\be
	\delta_\mfc = (\ddd \sigma_\mfc)_{e^Q} (e^Q \delta Y) \in \bR^n, 
	\quad
 	\delta_L = (\ddd \sigma_L)_{e^Q} (e^Q \delta Y) \in \mfm^\perp, 
 	\quad
 	\delta_R = (\ddd \sigma_R)_{e^Q} (e^Q \delta Y) \in \mfk, 
\ee
and define $D_\mfc = \diag(\delta_\mfc, -\delta_\mfc) \in \mfa$.
Since $\check{G}$ is open, for small values of $\vert t \vert$ we 
have $e^Q e^{t \delta Y} \in \check{G}$, whence
\be
	e^Q e^{t \delta Y} 
	= \sigma_L(e^Q e^{t \delta Y}) 
		e^{\diag(\sigma_\mfc(e^Q e^{t \delta Y}), 
			-\sigma_\mfc(e^Q e^{t \delta Y}))} 
		\sigma_R(e^Q e^{t \delta Y})^{-1}.
\ee
By taking the derivative of the above equation at $t = 0$, we obtain
\be
	\delta Y = \cosh(\wad_Q) \delta_L - \sinh(\wad_Q) \delta_L 
			+ D_\mfc - \delta_R.
\ee
It follows that
$(\delta Y)_\mfa = D_\mfc$ and $(\delta Y)_\mfm = -(\delta_R)_\mfm$,
meanwhile for the off-diagonal components we get 
\be
	(\delta Y)_{\mfa^\perp} = - \sinh(\wad_Q) \delta_L, 
	\quad
	(\delta Y)_{\mfm^\perp} = \cosh(\wad_Q) \delta_L 
					- (\delta_R)_{\mfm^\perp}.
\ee
By solving this linear system for $D_\mfc$, $\delta_L$ and $\delta_R$, 
the lemma follows.
\end{proof}

To proceed further, let us note that
\be
	\check{S} = \{ V \in S \, | \, V_1 \neq 0, \ldots, V_n \neq 0 \}
	\label{S_check}
\ee
is an open and dense submanifold of $S$ (\ref{S}), which contains $E$ 
(\ref{E}). The map
\be
	\tau \colon \check{S} \rightarrow M,
	\quad
	V \mapsto 
	\diag
	\left(
		\frac{V_1}{\vert V_1 \vert}, \ldots, \frac{V_n}{\vert V_n \vert}, 
		\frac{V_1}{\vert V_1 \vert}, \ldots, \frac{V_n}{\vert V_n \vert}
	\right)
	\label{tau}
\ee
defined on $\check{S}$ is smooth, satisfying $\tau(E) = \bsone_N$. 
Utilizing the natural identification 
\be
	T_E S \cong 
	\{ \delta V \in \bC^N 
		\, | \,
		\bsC \delta V + \delta V = 0,
		\, (\delta V)^* E + E^* \delta V = 0 
	\},
	\label{T_E_S}
\ee
for the action of the derivative of $\tau$ on each tangent vector 
$\delta V \in T_E S$ we find 
\be
	(\ddd \tau)_E (\delta V) 
	= \ri \, \diag
	\left( 
		\IM((\delta V)_1), \ldots, \IM((\delta V)_n),
		\IM((\delta V)_1), \ldots, \IM((\delta V)_n)
	\right)
	\in \mfm. 
	\label{tau_derivative}
\ee

Turning to the principal $U(1)$-bundle $\xi$ (\ref{xi}), notice that $E$ 
projects onto $\xi(E)$. Therefore, we can find a smooth local section
\be
	\check{\cO} \ni \rho \mapsto \cV(\rho) \in S
	\label{cV}
\ee
of $\xi$, defined on some open neighborhood $\check{\cO} \subset \cO$
of $\xi(E)$, such that $\cV(\xi(E)) = E$. Moreover, by `shrinking' 
$\check{\cO}$ if necessary, we may assume that
\be
	\cV(\rho) \in \check{S} \quad (\forall \rho \in \check{\cO}).
\ee
In order to fix the range of the derivative of section $\cV$ at point 
$\xi(E)$, notice that the map
\be
	T_E S \times T_E S \ni (\delta V, \delta W) 
	\mapsto
	\RE 
	\left(
		(\delta V)^* \delta W
	\right)
	\in \bR
	\label{inner_product_on_T_E_S}
\ee
is an inner product on the tangent space $T_E S$ (\ref{T_E_S}). 
Therefore, we may assume that the derivative operator 
$(\ddd \cV)_{\xi(E)}$ maps $T_{\xi(E)} \cO$ into the orthogonal 
complement of the vertical subspace
\be
	\ker((\ddd \xi)_E) 
	= \bR \ri E
	= \{ x \ri E \in T_E S \, | \, x \in \bR \}.
	\label{xi_vertical_subspace}
\ee
This requirement amounts to the constraint
\be
	(\delta V)^* E = E^* \delta V
	\quad (\forall \delta V \in \ran((\ddd \cV)_{\xi(E)})).
\label{cV_derivative_normalization}
\ee 

Remembering the local section (\ref{sigma_section}), let us consider 
the smooth map
\be
	\gamma \colon \check{G} \times \check{\cO} \rightarrow \cO,
	\quad
	(y, \rho) \mapsto \sigma_L(y)^{-1} \rho \sigma_L(y).
	\label{gamma}
\ee
Notice that $\gamma (e^Q, \xi(E)) = \xi(E) \in \check{\cO}$.
Therefore, there are some open neighborhoods $\hat{G} \subset \check{G}$ 
of $e^Q$, and $\hat{\cO} \subset \check{\cO}$ of $\xi(E)$, such that
\be
	\gamma (y, \rho) \in \check{\cO}
	\quad
	(\forall (y, \rho) \in \hat{G} \times \hat{\cO}).
\ee
Having equipped with $\tau$, $\cV$ and $\gamma$, now we define their 
composition
\be
	m \colon \hat{G} \times \hat{\cO} \rightarrow M,
	\quad
	(y, \rho) \mapsto m(y, \rho) = \tau(\cV(\gamma(y, \rho))),
	\label{m}
\ee
which is a smooth map satisfying $m(e^Q, \xi(E)) = \bsone_N$.
To compute the derivatives of the matrix entries of the diagonal
matrix
\be
	m(y, \rho) 
	= \diag(m_1(y, \rho), \ldots, m_n(y, \rho),
			m_1(y, \rho), \ldots, m_n(y, \rho))
\label{m_matrix_entries}
\ee 
at point $(e^Q, \xi(E))$, for each $c \in \bN_n$ we introduce the 
column vector $F_c \in \bC^N$ with components 
\be
	(F_c)_a = - (F_c)_{n + a} = \delta_{c, a}
	\quad
	(a \in \bN_n),
	\label{F_c}
\ee
together with the Lie algebra element
\be
	\Xi_c 
	= \ri 
	\left(
		F_c E^* + E F_c^* - \frac{2}{n} E E^*
	\right) \in \mfk.
	\label{Xi_c}
\ee

\begin{LEMMA}
\label{lemma_m_derivative}
Take arbitrary vectors $\delta Y \in \mfg$ and $Z \in \mfk$; 
then for each $c \in \bN_n$ we have
\be
	(\ddd m_c)_{(e^Q, \xi(E))} (e^Q \delta Y \oplus [Z, \xi(E)])
	= -\frac{\ri}{4} 
	\left\langle 
		\Xi_c, Z + \sinh(\wad_Q)^{-1} (\delta Y)_{\mfa^\perp}  
	\right\rangle.
	\label{m_c_derivative}
\ee
\end{LEMMA}

\begin{proof}
From (\ref{m}) it is clear that
\be
	(\ddd m)_{(e^Q, \xi(E))} (e^Q \delta Y \oplus [Z, \xi(E)])
	= (\ddd \tau)_E 
		(\ddd \cV)_{\xi(E)} 
			(\ddd \gamma)_{(e^Q, \xi(E))} 
				(e^Q \delta Y \oplus [Z, \xi(E)]).
	\label{m_derivative_chain_rule}
\ee 
Upon introducing the shorthand notation
\be
	X = Z + \sinh(\wad_Q)^{-1} (\delta Y)_{\mfa^\perp} \in \mfk,
	\label{X_def}
\ee
from (\ref{gamma}) and Lemma \ref{lemma_sigma_derivative} it is immediate
that
\be
	(\ddd \gamma)_{(e^Q, \xi(E))} (e^Q \delta Y \oplus [Z, \xi(E)])	
	= [X, \xi(E)] \in T_{\xi(E)} \cO.
\ee
Next, let us consider the tangent vector
\be
	\delta V = (\ddd \cV)_{\xi(E)} 
				(\ddd \gamma)_{(e^Q, \xi(E))} 
					(e^Q \delta Y \oplus [Z, \xi(E)])
	\in T_E S.
	\label{m_delta_V}
\ee
Since $\cV$ (\ref{cV}) is a local section, we have 
$\xi \circ \cV = \Id_{\check{\cO}}$, which entails 
\be
	(\ddd \xi)_E \delta V
	= (\ddd \xi)_E (\ddd \cV)_{\xi(E)} [X, \xi(E)]
	= [X, \xi(E)].
\ee
On the other hand, notice that the vector $X E$ belongs to the tangent 
space $T_E S$ (\ref{T_E_S}). Moreover, we find easily that
\be
	(\ddd \xi)_E (X E) = [X, \xi(E)].
\ee
From the last two equations we conclude that $\delta V - X E$ belongs to 
the vertical subspace (\ref{xi_vertical_subspace}), whence
$\delta V = X E + x \ri E$ with some $x \in \bR$. However, the value 
of $x$ is uniquely determined by (\ref{cV_derivative_normalization}), 
from where we infer that
\be
	\delta V = X E - \frac{E^* X E}{N } E.
\ee
Now, by combining (\ref{m_derivative_chain_rule}) and (\ref{m_delta_V}), 
we see that
\be
	(\ddd m)_{(e^Q, \xi(E))} (e^Q \delta Y \oplus [Z, \xi(E)]) 
	= (\ddd \tau)_E \delta V. 
\ee
Therefore, recalling (\ref{tau_derivative}) and (\ref{m_matrix_entries}), 
for each $c \in \bN_n$ we can write
\be
	(\ddd m_c)_{(e^Q, \xi(E))} (e^Q \delta Y \oplus [Z, \xi(E)]) 
	= \ri \, \IM( (\delta V)_c)
	= \ri \, \IM((X E)_c) - \frac{\tr(X E E^*)}{N}.
	\label{ddd_m_c}
\ee
Utilizing (\ref{F_c}), notice that
\be
	\ri \, \IM((X E)_c)
	= \half \ri \, \IM(F_c^* X E)
	= \quarter (F_c^* X E + E^* X F_c)
	= \frac{1}{4 \ri} \tr(\ri(F_c E^* + E F_c^*) X).
	\label{im_X_E_c}
\ee
Plugging this formula into (\ref{ddd_m_c}), the lemma follows.
\end{proof}

Now we are in a position to construct an appropriate local extension of
the Lax operator $L$ (\ref{L}). For this, let us consider the open subset
\be
	\hat{\cP}^\ext = \hat{G} \times \mfg \times \hat{\cO} \subset \cP^\ext,
	\label{hat_cP_ext}
\ee
which clearly contains the reference point $u_0$ 
(\ref{u_0}). Recalling (\ref{sigma_section}) and (\ref{m}), on 
$\hat{\cP}^\ext$ we define the smooth function
\be
	\varphi \colon \hat{\cP}^\ext \rightarrow K,
	\quad
	(y, Y, \rho) \mapsto \varphi(y, Y, \rho) = \sigma_R (y) m(y, \rho).
	\label{varphi}
\ee
Notice that $\varphi(u_0) = \bsone_N$. Finally, let us consider 
the locally defined smooth function
\be
	\tilde{L} \colon \hat{\cP}^\ext \rightarrow \mfg,
	\quad
	(y, Y, \rho) 
	\mapsto 
	\tilde{L}(y, Y, \rho) 
	= \varphi(y, Y, \rho)^{-1} Y \varphi(y, Y, \rho).
	\label{tilde_L}
\ee

\begin{LEMMA}
\label{L_local_extension}
The $\mfg$-valued function $\tilde{L}$ (\ref{tilde_L}), defined on a small 
open neighborhood of the reference point $u_0$ (\ref{u_0}), is a 
local extension of $L$ (\ref{L}). More precisely, we have
\be
	\tilde{L} \circ \Upsilon^S \big{|}_{(\Upsilon^S)^{-1}(\hat{\cP}^\ext)}
	= L \circ \pi^S \big{|}_{(\Upsilon^S)^{-1}(\hat{\cP}^\ext)}.
	\label{tilde_L_and_L}
\ee 
\end{LEMMA}

\begin{proof}
Take an arbitrary point
$\tilde{z} 
= (\tilde{q}, \tilde{p}, (\tilde{\eta}_L, \tilde{\eta}_R) U(1)_*)
\in (\Upsilon^S)^{-1}(\hat{\cP}^\ext)$, and let
\be
	(\tilde{y}, \tilde{Y}, \tilde{\rho})
	= \Upsilon^S(\tilde{z})
	= (\tilde{\eta}_L e^{\tilde{Q}} \tilde{\eta}_R^{-1},
		\tilde{\eta}_R L(\tilde{q}, \tilde{p}) \tilde{\eta}_R^{-1},
		\tilde{\eta}_L \xi(E) \tilde{\eta}_L^{-1})
	\in \hat{\cP}^\ext.
	\label{tilde_y_Y_rho}
\ee
By applying the local section $\sigma$ (\ref{sigma_section}) on the Lie 
group element $\tilde{y} \in \hat{G} \subset \check{G}$, we see that
\be
	\tilde{y} 
	= \sigma_L (\tilde{y})
		e^{\diag(\sigma_\mfc(\tilde{y}), -\sigma_\mfc(\tilde{y}))}
		\sigma_R(\tilde{y})^{-1}.
\ee
Recalling $\pi$ (\ref{pi}), it is immediate that
$\tilde{q} = \sigma_\mfc(\tilde{y})$ and
$(\tilde{\eta}_L, \tilde{\eta}_R) M_* 
= (\sigma_L(\tilde{y}), \sigma_R(\tilde{y})) M_*$.
Therefore, we have
\be
	\tilde{\eta}_L = \sigma_L(\tilde{y}) \cD
	\quad \text{and} \quad
	\tilde{\eta}_R = \sigma_R(\tilde{y}) \cD
	\label{eta_and_sigma}
\ee
with some diagonal matrix $\cD \in M$.

Next, we inspect the Lie group element
$m(\tilde{y}, \tilde{\rho}) 
= \tau(\cV(\gamma(\tilde{y}, \tilde{\rho}))) \in M$.
Remembering (\ref{gamma}), notice that
\be
	\gamma(\tilde{y}, \tilde{\rho})
	= \sigma_L(\tilde{y})^{-1} \tilde{\rho} \sigma_L(\tilde{y})
	= \sigma_L(\tilde{y})^{-1} \tilde{\eta}_L 
		\xi(E) \tilde{\eta}_L^{-1} \sigma_L(\tilde{y})
	= \cD \xi(E) \cD^{-1} 
	= \xi(\cD E).
\ee 
Since $\cV$ (\ref{cV}) is a local section of $\xi$ (\ref{xi}), we find 
that
\be
	\xi(\cV(\gamma(\tilde{y}, \tilde{\rho})))
	= (\xi \circ \cV)(\xi(\cD E)) = \xi(\cD E).
\ee
Therefore, there is some $\psi \in \bR$, such that 
$\cV(\gamma(\tilde{y}, \tilde{\rho})) = e^{\ri \psi} \cD E$.
Recalling (\ref{tau}), it follows that
\be
	m(\tilde{y}, \tilde{\rho}) 
	= \tau(e^{\ri \psi} \cD E)
	= e^{\ri \psi} \cD.
	\label{m_tilde}
\ee
Due to relationships (\ref{eta_and_sigma}) and 
(\ref{m_tilde}) we observe that
\be
	\varphi(\tilde{y}, \tilde{Y}, \tilde{\rho})
	= \sigma_R (\tilde{y}) m(\tilde{y}, \tilde{\rho}) 
	=  e^{\ri \psi} \tilde{\eta}_R.
\ee
Therefore, recalling (\ref{tilde_L}) and (\ref{tilde_y_Y_rho}), we find 
immediately that
\be
	\tilde{L}(\Upsilon^S(\tilde{z}))
	= \varphi(\tilde{y}, \tilde{Y}, \tilde{\rho})^{-1} 
		\tilde{Y} \varphi(\tilde{y}, \tilde{Y}, \tilde{\rho})
	= L(\tilde{q}, \tilde{p})
	= L (\pi^S(\tilde{z})).
\ee
Since $\tilde{z}$ is an arbitrary element of 
$(\Upsilon^S)^{-1}(\hat{\cP}^\ext)$, the lemma follows.
\end{proof}

\subsection{Computing the $r$-matrix}
Let us choose some dual bases $\{ T_A \}$, $\{ T^A \}$ in $\mfg$, i.e. 
$\langle T^A, T_B \rangle = \delta^A_B$, and consider 
the function
\be
	\cP^\ext \ni (y, Y, \rho) \mapsto \bsY(y, Y, \rho) = Y \in \mfg,
	\label{bsY}
\ee
together with its components
\be
	\cP^\ext \ni (y, Y, \rho) 
	\mapsto 
	\bsY^A (y, Y, \rho) = \langle T^A, Y \rangle \in \bR,
	\label{bsY_A}
\ee 
defined on the extended phase space. Notice that
$\bsY = \sum_A \bsY^A T_A$.
Since $\bsY^A$ depends only on variable $Y$, its only nontrivial gradient 
(\ref{gradient}) is 
\be
	\nabla^\mfg \bsY^A \equiv T^A, 
	\label{nabla_bsY_A}
\ee
whence from (\ref{PB_ext}) we obtain that
$\{ \bsY^A, \bsY^B \}^\ext = - \langle [ T^A, T^B ], \bsY \rangle$.
As usual in the theory of integrable systems, these relationships can be 
succinctly rewritten in the standard St Petersburg tensorial\footnote{ 
We write 
$L_1 = L \otimes \bsone$, 
$L_2 = \bsone \otimes L$,
together with  
$r_{12} = \sum r^{A, B} T_A \otimes T_B$,
$r_{21} = \sum r^{A, B} T_B \otimes T_A$, etc. 
} 
notation. Indeed, upon introducing the 
quadratic Casimir
\be
	\Omega_{12} = \Omega_{21}
	= \sum_A T_A \otimes T^A \in \mfg \otimes \mfg,
	\label{Casimir}
\ee
we find easily that
\be
	\{ \bsY_1, \bsY_2 \}^\ext 
	= \sum_{A, B} \{ \bsY^A, \bsY^B \}^\ext T_A \otimes T_B  	
	= [ - \Omega_{12} / 2, \bsY_1 ] 
		- [ -\Omega_{21}/2, \bsY_2 ].
	\label{bsY_tensorial_PB}
\ee
Now, from (\ref{tilde_L}) and (\ref{bsY}) we see that $\tilde{L}$
can be obtained from $\bsY$ by the gauge transformation
\be
	\tilde{L} = \varphi^{-1} \bsY \varphi.
	\label{tilde_L_vs_bsY}
\ee
It readily follows\footnote{
For details on gauge transformations see e.g. \cite{Babelon_Viallet_1990}, 
or Chapter 2 in \cite{Babelon_Bernard_Talon_book}.
} 
that
$\{ \tilde{L}_1, \tilde{L}_2 \}^\ext
= [ \tilde{r}_{12}, \tilde{L}_1 ]
- [ \tilde{r}_{21}, \tilde{L}_2 ]$ with the transformed $r$-matrix
\be
	\tilde{r}_{12}
	= \varphi_1^{-1} \varphi_2^{-1}
	\left(
		-\half \Omega_{12}
		- \{ \varphi_1, \bsY_2 \}^\ext \varphi_1^{-1}
		+ \half 
			\left[ 
				\{ \varphi_1, \varphi_2 \}^\ext \varphi_1^{-1} \varphi_2^{-1},
				 \bsY_2
			\right]
	\right)
	\varphi_1 \varphi_2.
	\label{tilde_r}
\ee
Recalling (\ref{PB_S_vs_PB_red}) and (\ref{red_PB_vs_ext_PB}), it is thus
immediate that for the Lax matrix $L$ (\ref{L}) we have
\be
	\{ L_1, L_2 \}^S(q, p)
	= [ r_{12}(q, p), L_1(q, p) ]
		- [ r_{21}(q, p), L_2(q, p) ]
\label{L_tensorial_PB}
\ee
with the $r$-matrix
\be
	r_{12}(q, p) 
	= 2 \tilde{r}_{12}(u_0) 
	= - \Omega_{12}
		- 2 \{ \varphi_1, \bsY_2 \}^\ext (u_0)  
		+ 
		\left[ 
			\{ \varphi_1, \varphi_2 \}^\ext (u_0),
				L_2(q, p)
		\right].
\label{r}
\ee
In the following we use extensively the special basis of $\mfg$
introduced in the appendix. As a first step, we define the Lie algebra 
elements
\be
	Z_{e_a \pm e_b} = \frac{D^+_a + D^+_b}{\sqrt{2}} \in \mfm
	\quad \text{and} \quad
	Z_{2 e_c} = D^+_c \in \mfm,
\label{Z_alpha}
\ee
where $a, b, c \in \bN_n$ and $a < b$. Observe that these matrices
are labeled by the $C_n$-type positive roots (\ref{calR_+}). Now, we can 
formulate the main result of the paper. 

\begin{THEOREM}
\label{theorem_r_matrix}
The Lax matrix (\ref{L}) of the classical hyperbolic $BC_n$ Sutherland 
model verifies the $r$-matrix Poisson bracket (\ref{L_tensorial_PB}) 
with the $q$-dependent $r$-matrix
\be
\begin{split}
	r_{12}(q) 
	= & \, 2 \sum_{\alpha, \epsilon} 
			\coth(\alpha(q)) X^{+, \epsilon}_\alpha 
				\otimes X^{-, \epsilon}_\alpha
	-2 \sum_\alpha \frac{1}{\sinh(\alpha(q))} 
			Z_\alpha \otimes X^{-, \ri}_\alpha \\
	& - \sum_c (D^+_c \otimes D^+_c + D^-_c \otimes D^-_c)
		- \sum_{\alpha, \epsilon} 
			(X^{+, \epsilon}_\alpha \otimes X^{+, \epsilon}_\alpha
			+X^{-, \epsilon}_\alpha \otimes X^{-, \epsilon}_\alpha).
\end{split}
\label{r_matrix}
\ee
\end{THEOREM}

\begin{proof}
As dictated by (\ref{r}), we inspect the formula of $r_{12}(q, p)$
term-by-term. Using the basis given in the appendix, it is
clear that the quadratic Casimir (\ref{Casimir}) takes the form
\be
	\Omega_{12}
	= \sum_c (D^-_c \otimes D^-_c - D^+_c \otimes D^+_c)
	+ \sum_{\alpha, \epsilon} 
		(X^{-, \epsilon}_\alpha \otimes X^{-, \epsilon}_\alpha
		-X^{+, \epsilon}_\alpha \otimes X^{+, \epsilon}_\alpha).
\ee
Recalling (\ref{PB_ext}) and (\ref{varphi}), it is also clear that
\be
	\{ \varphi_1, \bsY_2 \}^\ext (u_0) 
	= \sum_A
		\left(
			(\ddd \sigma_R)_{e^Q} (e^Q T^A) 
			+ (\ddd m)_{(e^Q, \xi(E))} (e^Q T^A \oplus 0)
		\right)
		\otimes T_A.
\label{varphi_Y}
\ee
Now, from Lemma \ref{lemma_sigma_derivative} it is immediate that
\be
	\sum_A (\ddd \sigma_R)_{e^Q} (e^Q T^A) \otimes T_A
	= \sum_{c} D^+_c \otimes D^+_c
		+ \sum_{\alpha, \epsilon}
			X^{+, \epsilon}_\alpha \otimes 
			\left(
				X^{+, \epsilon}_\alpha 
				- \coth(\alpha(q)) X^{-, \epsilon}_\alpha
			\right),
\label{varphi_Y_1}
\ee
whereas Lemma \ref{lemma_m_derivative} leads to the formula
\be
	\sum_A (\ddd m)_{(e^Q, \xi(E))} (e^Q T^A \oplus 0) \otimes T_A 
	= - \frac{\sqrt{2}}{4} \sum_A \sum_{c}
		\left\langle 
			\Xi_c, \sinh(\wad_Q)^{-1} (T^A)_{\mfa^\perp} 
		\right\rangle
		D^+_{c} \otimes T_A. 
\ee
Remembering definition (\ref{Xi_c}), notice that the $c$-dependent part 
of $\Xi_c$ can be rewritten as
\be
\ri (F_c E^* + E F_c^*)
= 2 \sqrt{2} (D^+_c + X^{+, \ri}_{2 e_c})
	+ 2 \sum_{d = 1}^{c - 1} 
		(X^{+, \ri}_{e_d - e_c} + X^{+, \ri}_{e_d + e_c})  
	+ 2 \sum_{d = c + 1}^{n} 
		(X^{+, \ri}_{e_c - e_d} + X^{+, \ri}_{e_c + e_d}).  
\ee 
Therefore, upon introducing the Lie algebra element
\be
	\Psi(q) 
	= \frac{1}{N} \sum_{A} 
		\left \langle
			\ri E E^*, \sinh(\wad_Q)^{-1} (T^A)_{\mfa^\perp}
		\right \rangle
		T_A \in \mfg,
\label{Psi}
\ee
we find easily that
\be
	\sum_A (\ddd m)_{(e^Q, \xi(E))} 
		(e^Q T^A \oplus 0) \otimes T_A  
	= \ri \bsone_N \otimes \Psi(q) 
	+ \sum_\alpha \frac{1}{\sinh(\alpha(q))} 
		Z_\alpha \otimes X^{-, \ri}_\alpha.	
\label{varphi_Y_2}
\ee
By plugging formulae (\ref{varphi_Y_1}) and (\ref{varphi_Y_2}) into 
(\ref{varphi_Y}), the control over 
$\{ \varphi_1, \bsY_2 \}^\ext(u_0)$ is complete.

To proceed further, let us introduce the locally defined smooth functions
\be
	\tilde{\sigma}_R (y, Y, \rho) = \sigma_R(y)
	\quad \text{and} \quad 
	\tilde{m} (y, Y, \rho) = m(y, \rho)
	\quad
	((y, Y, \rho) \in \hat{\cP}^\ext). 
\ee
Due to (\ref{varphi}) it is clear that 
$\varphi = \tilde{\sigma}_R \tilde{m}$. 
Now, from (\ref{PB_ext}) we see that on $\hat{\cP}^\ext$ we have
\be
	\{ (\tilde{\sigma}_R)_1, (\tilde{\sigma}_R)_2 \}^\ext \equiv 0,
	\quad
	\{ (\tilde{\sigma}_R)_1, (\tilde{m})_2 \}^\ext \equiv 0,
	\quad
	\{ (\tilde{m})_1, (\tilde{\sigma}_R)_2 \}^\ext \equiv 0,
\ee
therefore 
$\{ \varphi_1, \varphi_2 \}^\ext(u_0) 
= \{ (\tilde{m})_1, (\tilde{m})_2 \}^\ext(u_0)$
readily follows. Keeping our focus on this relationship, from 
(\ref{gradient_def}) and Lemma \ref{lemma_m_derivative} it is
immediate that 
\be
	\nabla^\cO \RE(\tilde{m}_c) (u_0) = 0
	\quad \text{and} \quad
	\nabla^\cO \IM(\tilde{m}_c) (u_0) = - \Xi_c / 4.
\ee
Note that the Lie algebra element $\Xi_c$ (\ref{Xi_c}) can be represented 
as an appropriate commutator. Indeed, upon introducing the matrices
\be
	A_c = \frac{1}{n} \sum_{d = 1}^n (e_{c, d} - e_{d, c}) \in \mfu(n)
	\quad \text{and} \quad
	V_c = \diag(A_c, A_c) \in \mfk,
\label{A_c_V_c}
\ee
we find immediately that $\Xi_c = \mu^{-1} [V_c, \xi(E)]$.
Therefore, from (\ref{PB_ext}) and (\ref{omega_cO}) we obtain that
\be
\begin{split}
	\{ \tilde{m}_c, \tilde{m}_d \}^\ext(u_0)
	= -\frac{1}{16 \mu^2} 
		\left \langle \xi(E), [V_c, V_d] \right \rangle 
	= 0, 
\end{split}
\ee
for all $c, d \in \bN_n$. Thus, we end up with the simple relationship
$\{ \varphi_1, \varphi_2 \}^\ext (u_0) = 0$. 	

We conclude the proof with the observation that the term 
$\ri \bsone_N \otimes \Psi(q)$ appearing in (\ref{varphi_Y_2}) can be 
neglected, since it commutes with $L_1(q, p)$.
Therefore, by simply plugging the above derived formulae into (\ref{r}),
the theorem follows.
\end{proof}

Switching to the standard basis $\{ e_{k, l} \}$ of the matrix Lie
algebra $\mfgl(N, \bC)$, the $r$-matrix (\ref{r_matrix}) can be rewritten 
as
\be
\begin{split}
	r_{12}(q)
	= & \, \sum_{\substack{a, b = 1 \\ (a \neq b)}}^n
			\coth(q_a - q_b) 
			(e_{a, b} + e_{n + a, n + b}) 
			\otimes 
			(e_{b, a} - e_{n + b, n + a}) \\
	& + \sum_{a, b = 1}^n
		\coth(q_a + q_b) 
		(e_{a, n + b} + e_{n + a, b}) 
		\otimes 
		(e_{n + b, a} - e_{b, n + a}) \\
	& + \half \sum_{\substack{a, b = 1 \\ (a \neq b)}}^n
			\frac{1}{\sinh(q_a - q_b)} 
			(e_{a, a} + e_{n + a, n + a} 
				+ e_{b, b} + e_{n + b, n + b}) 
			\otimes 
			(e_{a, b} - e_{n + a, n + b}) \\
	& - \half \sum_{a, b = 1}^n
			\frac{1}{\sinh(q_a + q_b)} 
			(e_{a, a} + e_{n + a, n + a} 
				+ e_{b, b} + e_{n + b, n + b}) 
			\otimes 
			(e_{a, n + b} - e_{n + a, b}) \\
	& + \sum_{a, b = 1}^n 
		\left(
			e_{a, b} \otimes e_{n + b, n + a}
			+ e_{n + a, n + b} \otimes e_{b, a}
			+ e_{a, n + b} \otimes e_{b, n + a}
			+ e_{n + a, b} \otimes e_{n + b, a}
		\right).
\end{split}
\label{r_in_standard_basis}
\ee
To conclude this subsection, notice that the above $r$-matrix can be 
seen as a generalization of the $C_n$-type $r$-matrix constructed by 
Avan, Babelon and Talon. Indeed, up to a constant conjugation, 
the $q$-dependent part of (\ref{r_in_standard_basis}) can be identified 
with the $r$-matrix of the $C_n$ Sutherland model presented in 
\cite{Avan_Babelon_Talon_1994}. Nevertheless, as one can easily verify 
by inspecting the $r$-matrix Poisson bracket (\ref{L_tensorial_PB}), 
in the special case $\kappa = 0$ the $q$-independent part of 
(\ref{r_in_standard_basis}) can be safely neglected. In other words, 
with the specialization $\kappa = 0$ we can also recover the 
$C_n$-type $r$-matrix of paper \cite{Avan_Babelon_Talon_1994}. 

\subsection{Lax representation of the dynamics}
Having constructed an $r$-matrix for the $BC_n$ Sutherland model, we can
automatically provide a Lax representation for the dynamics as well. For 
this, we need the operator version of $r_{12}(q)$ (\ref{r_matrix}),
which is defined via the natural identifications 
\be
	\mfg \otimes \mfg \cong \mfg \otimes \mfg^* \cong \End(\mfg).
\ee 
More precisely, the linear operator $R(q) \in \End(\mfg)$ corresponding 
to the element $r_{12}(q) \in \mfg \otimes \mfg$ can be recovered from 
the formula
\be
	R(q) Y = \tr_2 (r_{12}(q) Y_2)
	\quad
	(Y \in \mfg),
\label{R_operator}
\ee
where the linear operator $\tr_2$ defined by 
$\tr_2(X \otimes Y) = \tr(Y) X$ is the usual partial trace on the second 
factor. From (\ref{r_matrix}) and (\ref{R_operator}) it is immediate that 
\be
	R(q) Y 
	= 2 \sum_{\alpha, \epsilon} \coth(\alpha(q))
			\langle X^{-, \epsilon}_\alpha, Y \rangle 
			X^{+, \epsilon}_\alpha
	- 2 \sum_{\alpha} \frac{1}{\sinh(\alpha(q))} 
			\langle X^{-, \ri}_\alpha, Y \rangle Z_\alpha
	- Y^*. 
\label{R}
\ee

Now, let us introduce the matrix-valued function
$B = \half \left( L + R L \right)$
defined on the phase space $\cP^S$. Since the Lax matrix
$L$ (\ref{L}) can be written as
\be
\begin{split}
	L(q, p) = & \sqrt{2} 
			\sum_{c = 1}^n p_c D^-_c 
				-\sqrt{2} \sum_{c = 1}^n
					\frac{\nu + \kappa \cosh(2 q_c)}
					{\sinh(2 q_c)}  
					X^{-, \ri}_{2 e_c} 
	\\	
		& -2 \mu \sum_{1 \leq a < b \leq n}
			\left(
				\frac{X^{-, \ri}_{e_a - e_b}}
				{\sinh(q_a - q_b)} 
				+ \frac{X^{-, \ri}_{e_a + e_b}}
				{\sinh(q_a + q_b)} 
			\right)
		-\kappa \ri \bsC,	
\end{split}
\label{L_in_ONB}
\ee
from (\ref{R}) it follows that $B$ has the block matrix structure
\be
	B = 
	\begin{bmatrix}
		S & T \\
		T & S
	\end{bmatrix},
\label{B_block_matrix}
\ee
where $S$ and $T$ are appropriate $\mfu(n)$-valued functions 
on $\cP^S$. Namely, for their matrix entries we have
\be
	T_{c, c} = \ri \frac{\nu \cosh(2 q_c) + \kappa}{\sinh^2(2 q_c)},
	\quad
	T_{a, b} = \ri \mu \frac{\cosh(q_a + q_b)}{\sinh^2(q_a + q_b)},
	\quad
	S_{a, b} = - \ri \mu \frac{\cosh(q_a - q_b)}{\sinh^2(q_a - q_b)},
\ee
meanwhile
\be
	S_{c, c} 
	= \ri \frac{\nu + \kappa \cosh(2 q_c)}{\sinh^2(2 q_c)}
	+ \ri \mu \sum_{\substack{d = 1 \\ (d \neq c)}}^n 
					\left(
						\frac{1}{\sinh^2(q_c - q_d)}
						+ \frac{1}{\sinh^2(q_c + q_d)}
					\right),
\ee
where $a, b, c \in \bN_n$ and $a \neq b$. Notice that $B$ is actually 
a $\mfk$-valued map depending only on $q$.

As we have discussed in Section \ref{S2}, the reduced Hamiltonian 
corresponding to $F_2$ (\ref{F_2}) coincides 
with the Hamiltonian of the Sutherland model (\ref{H_S}), i.e.
$H^S = \langle L, L \rangle / 4$. 
By applying the Hamiltonian vector field $\bsX_{H^S} \in \mfX(\cP^S)$
on $L$, from the $r$-matrix Poisson bracket (\ref{L_tensorial_PB}) we 
obtain 
\be
	\bsX_{H^S} [L] 	= \half [R L, L] = [B, L].
\ee 
That is, along each trajectory $t \mapsto (q(t), p(t))$ of the
Sutherland dynamics the Lax equation
\be
	\dot{L} = [B, L]
\label{Lax_eqn}
\ee
holds. The above observation can be sharpened as follows.

\begin{THEOREM}
\label{theorem_Lax_eqn}
A smooth curve in the phase space $\cP^S$ (\ref{cP_S}) is an integral 
curve of the hyperbolic $BC_n$ Sutherland dynamics, if and only if, 
along the curve the Lax equation (\ref{Lax_eqn}) is satisfied.
\end{THEOREM}

\begin{proof}
By applying repeatedly the identity
\be
	\frac{\cosh(x)}{\sinh^2(x)} \frac{1}{\sinh(y)}
	- \frac{\cosh(y)}{\sinh^2(y)} \frac{1}{\sinh(x)}
	= \frac{1}{\sinh(x + y)}
		\left(
			\frac{1}{\sinh^2(x)} - \frac{1}{\sinh^2(y)}
		\right),
\label{hyperbolic_identity}
\ee
elementary algebraic manipulations lead to the formula
\be
\begin{split}
	[B, L] 
	= & \, 2 \mu \sum_{1 \leq a < b \leq n}
		\left(
			(p_a - p_b) 
			\frac{\cosh(q_a - q_b)}{\sinh^2(q_a - q_b)} 
			X^{-, \ri}_{e_a - e_b}
			+ (p_a + p_b) 
			\frac{\cosh(q_a + q_b)}{\sinh^2(q_a + q_b)} 
			X^{-, \ri}_{e_a + e_b}
		\right) \\
	& + 2 \sqrt{2} \sum_{c = 1}^n p_c 
			\frac{\nu \cosh(2 q_c) + \kappa}{\sinh^2(2 q_c)} 
				X^{-, \ri}_{2 e_c} 
	- \sqrt{2} \sum_{c = 1}^n \frac{\partial H}{\partial q_c} D^-_c.
\end{split}
\label{B_L_commutator}
\ee
On the other hand, by differentiating $L$ (\ref{L_in_ONB}) 
along an arbitrary smooth curve $(q(t), p(t)) \in \cP^S$
with respect to time $t$, we find easily that
\be
\begin{split}
	\dot{L} 
	= & \, 2 \mu \sum_{1 \leq a < b \leq n}
		\left(
			(\dot{q}_a - \dot{q}_b) 
			\frac{\cosh(q_a - q_b)}{\sinh^2(q_a - q_b)} 
			X^{-, \ri}_{e_a - e_b}
			+ (\dot{q}_a + \dot{q}_b) 
			\frac{\cosh(q_a + q_b)}{\sinh^2(q_a + q_b)} 
			X^{-, \ri}_{e_a + e_b}
		\right) \\
	& + 2 \sqrt{2} \sum_{c = 1}^n \dot{q}_c 
			\frac{\nu \cosh(2 q_c) + \kappa}{\sinh^2(2 q_c)} 
				X^{-, \ri}_{2 e_c} 
	+ \sqrt{2} \sum_{c = 1}^n \dot{p}_c D^-_c.
\end{split}
\label{L_dot}
\ee
Hence, by comparing formulae (\ref{B_L_commutator}) and (\ref{L_dot}), 
we conclude that the Lax equation is equivalent to the Hamiltonian 
equation of motion of the Sutherland model.
\end{proof}

\section{Discussion}
\label{S4}
\setcounter{equation}{0}
Starting with the seminal paper \cite{Avan_Talon_1993},
a lot of effort has been devoted to explore the $r$-matrix structure
of the Calogero--Moser--Sutherland many-particle systems. In
this paper we contribute to this research area by constructing a 
dynamical $r$-matrix for the hyperbolic $BC_n$ Sutherland model with 
three independent parameters (\ref{H_S}). The outcome of our analysis 
is consistent with the results of \cite{Avan_Babelon_Talon_1994} on 
the $r$-matrix structure of the $C_n$ Sutherland model with two 
independent coupling constants. We wish to mention that the authors 
of paper \cite{Forger_Winterhalder_2002} have also constructed a 
dynamical $r$-matrix for a restricted class of the $BC_n$-type Sutherland 
models. More precisely, their results are valid under the same 
restriction on the coupling parameters that was sticked to these models
in the fundamental work of Olshanetsky and Perelomov \cite{OlshaPere76}.
Recall also that the $BC_n$-type $r$-matrix in 
\cite{Forger_Winterhalder_2002}
explicitly depends on the coupling parameters. Note, however, that our 
$r$-matrix (\ref{r_matrix}) is \emph{independent} of the coupling 
parameters, and so it is equally valid for the $B_n$, $C_n$ and $BC_n$ 
Sutherland models, too. This `universal' feature of (\ref{r_matrix}) 
naturally indicates a Yang--Baxter-type algebraic structure 
behind these models. We wish to investigate this important topic in  
future publications. A related open problem is to explore the 
relationship between our $r$-matrix and the $BC_n$-type Sutherland 
model with two types of particles (see e.g. \cite{Hashizume}, 
\cite{Ayadi_Feher_2011}). 

Similar questions arise in the context of the elliptic Calogero models,
too. We have a fairly complete understanding of the $r$-matrix
structure of Krichever's spectral parameter dependent Lax matrix 
\cite{Krichever_1980} for the $A_n$-type model (see \cite{Sklyanin_1994}, 
\cite{Braden_Suzuki_1994}). These elliptic $r$-matrices are dynamical
objects, depending on the particle coordinates. As is known, one can 
even construct a Lax matrix for the elliptic $A_n$-type model, which 
obeys an $r$-matrix Poisson bracket with Belavin's \cite{Belavin_1981} 
\emph{non-dynamical} elliptic $r$-matrix. (The details on the elliptic 
case can be found in \cite{Hou_Wang_1999},
whereas \cite{Feher_Pusztai_2000} contains an elementary account on the 
non-dynamical $r$-matrix structure of the degenerate $A_n$-type models.) 
However, for the $r$-matrix structure of the elliptic $BC_n$ Calogero 
model only partial results are known. Namely, in 
\cite{Forger_Winterhalder_2002} a dynamical $r$-matrix is constructed
for the elliptic $BC_n$ Calogero model with the aforementioned 
restriction on coupling constants. For this restricted class of $BC_n$ 
models the Lax representation with non-dynamical $r$-matrix has been
also investigated (see \cite{Forger_Winterhalder_2003}).
Nevertheless, to our knowledge, the $r$-matrix of the most general three
parameter dependent elliptic $BC_n$ Calogero model is still missing.
It also appears to be an interesting open problem to provide 
$r$-matrices for the universal Lax operators constructed in 
\cite{Bordner_Corrigan_Sasaki}.
An equally ambitious project would be to construct an $r$-matrix for 
Inozemtsev's \cite{Inozemtsev_1989} many-parameter dependent elliptic 
model, too. We hope that appropriate generalizations of our $r$-matrix 
(\ref{r_matrix}) may play a role in clarifying these issues. 

To conclude the paper, let us recall that
the Ruijsenaars--Schneider--van Diejen (RSvD) models 
(see e.g. \cite{RuijSchneider}, \cite{vanDiejen1994}) 
are natural generalizations of the Calogero--Moser--Sutherland (CMS) 
particle systems. The $r$-matrix structure of the $A_n$-type 
Ruijsenaars--Schneider models is well understood (for details on the 
elliptic models see e.g. \cite{Nijhoff_et_al_1996},
\cite{Hou_Wang_2000}), but for the generic non-$A_n$-type models even
the Lax representation of the dynamics is missing. Quite surprisingly,
the construction of a Lax matrix for the rational $BC_n$ 
RSvD model with three independent
coupling parameters was carried out only in the recent paper 
\cite{Pusztai_NPB2012}. Due to the dual reduction picture presented
in \cite{Pusztai_NPB2012}, we expect that the $r$-matrix structure
of the rational $BC_n$ RSvD model can be analyzed by the same techniques 
we outlined in Section \ref{S3}. As for the $A_n$-type systems, it has 
been observed that in some sense the CMS and the RSvD models can be 
characterized by the \emph{same} $r$-matrices (for details see 
\cite{Suris_1997}). It appears to be an interesting
question whether the dual reduction picture behind the CMS and
the RSvD models can provide a geometric explanation of this
remarkable phenomenon. We wish to come back to these problems in 
future publications.

\renewcommand{\thesection}{A}
\section{Convenient basis for $\mfu(n, n)$}
\label{appendix}
\renewcommand{\theequation}{A.\arabic{equation}}
\setcounter{equation}{0}
As a supplementary material to the main text, in this appendix we present
a convenient basis for the real Lie algebra $\mfg = \mfu(n, n)$ adapted 
to the orthogonal decomposition (\ref{mfg_refined_decomposition}).
First, for each $c \in \bN_n$ we define the linear functional
\be
	e_c \colon \bR^n \rightarrow \bR,
	\quad
	q = (q_1, \ldots, q_n) \mapsto e_c(q) = q_c.
\label{e_c}
\ee 
Clearly the set of functionals
\be
	\calR_+ 
	= \{ e_a \pm e_b \, | \, 1 \leq a < b \leq n \}
	\cup
	\{2 e_c \, | \, c \in \bN_n \}	
\label{calR_+}
\ee
can be seen as a family of \emph{positive} roots of type $C_n$. 
We also need the standard $N \times N$ elementary matrices $e_{k, l}$. 
Recall that for their matrix entries we have 
$(e_{k, l})_{k', l'} = \delta_{k, k'} \delta_{l, l'}$.

Now, for each $c \in \bN_n$ we define the diagonal matrices
\be
	D_c^+ = \frac{\ri}{\sqrt{2}} (e_{c, c} + e_{n + c, n + c}),
	\quad
	D_c^- = \frac{1}{\sqrt{2}} (e_{c, c} - e_{n + c, n + c}).
\label{D+-}
\ee 
Clearly $\{ D_c^+ \}$ is a basis in $\mfm$, whereas $\{ D_c^-\}$ is 
basis in $\mfa$, satisfying the relations
\be
	\langle D_c^+, D_d^+ \rangle = - \delta_{c, d},
	\quad
	\langle D_c^-, D_d^- \rangle = \delta_{c, d}.
\label{D_scalar_product}
\ee
Next, for each $c \in \bN_n$ we introduce the matrices
\be
	X_{2 e_c}^{\pm, \ri} 
	= -\frac{\ri}{\sqrt{2}} (e_{c, n + c} \pm e_{n + c, c}).
\ee
Also, for all $1 \leq a < b \leq n$ we define the following
matrices with purely real entries
\be
\begin{split}
	X_{e_a - e_b}^{\pm, \rr} &
	= \half (e_{a, b} \mp e_{b, a} 
				\pm e_{n + a, n + b} - e_{n + b, n + a}), \\
	\quad
	X_{e_a + e_b}^{\pm, \rr} &
	= -\half (e_{a, n + b} - e_{b, n + a} 
				\pm e_{n + a, b} \mp e_{n + b, a}), 
\end{split}
\ee
together with the following matrices with purely imaginary entries
\be
\begin{split}
	X_{e_a - e_b}^{\pm, \ri} &
	= \frac{\ri}{2} (e_{a, b} \pm e_{b, a} 
					\pm e_{n + a, n + b} + e_{n + b, n + a}), \\
	\quad
	X_{e_a + e_b}^{\pm, \ri} &
	= - \frac{\ri}{2} (e_{a, n + b} + e_{b, n + a} 
					\pm e_{n + a, b} \pm e_{n + b, a}). 
\end{split}
\ee
The set of vectors $\{ X_\alpha^{+, \epsilon} \}$ forms a basis in 
$\mfm^\perp$, meanwhile $\{ X_\alpha^{-, \epsilon} \}$ is a basis in 
$\mfa^\perp$. Note that
\be
	\langle X^{+, \epsilon}_\alpha, X^{+, \epsilon'}_{\alpha'} \rangle
	= - \delta_{\alpha, \alpha'} \delta_{\epsilon, \epsilon'},
	\quad 
	\langle X^{-, \epsilon}_\alpha, X^{-, \epsilon'}_{\alpha'} \rangle
	= \delta_{\alpha, \alpha'} \delta_{\epsilon, \epsilon'}.
\label{X_scalar_product}
\ee
Due to the orthogonality relations (\ref{D_scalar_product}) and 
(\ref{X_scalar_product}), the construction of the corresponding dual 
basis is trivial. Keeping in mind the notation introduced in (\ref{Q}), 
it is worth mentioning that the above listed vectors satisfy the 
commutation relations
\be
	[Q, X_\alpha^{\pm, \epsilon}] 
	= \alpha(q) X_\alpha^{\mp, \epsilon}, 
\ee
where $q \in \bR^n$, $\alpha \in \calR_+$ and 
$\epsilon \in \{ \rr, \ri \}$.

\medskip
\noindent
\textbf{Acknowledgments.}
We wish to thank L.~Feh\'er (University of Szeged) for useful comments
on the manuscript. This work has been supported in part by the Hungarian 
Scientific Research Fund (OTKA) under grant K 77400.


\end{document}